\documentclass[11pt]{article}
\usepackage[utf8]{inputenc}
\usepackage[english]{babel}
\usepackage{amsfonts}
\usepackage{amsmath}
\usepackage{amssymb}
\usepackage{amsthm}
\usepackage{tikz}
\usepackage{pgfpages}
\usepackage{comment}

\usetikzlibrary{arrows}
\makeatother

\usepackage{algorithm2e}

\usepackage[margin=1in]{geometry}
\theoremstyle{plain}
\newtheorem{thm}{\protect\theoremname}
\theoremstyle{remark}

\theoremstyle{definition}

\newtheorem{lemma}[thm]{Lemma}
\newtheorem{corollary}{Corollary}
\newtheorem{definition}{Definition}
\newtheorem{theorem}{Theorem}

\makeatletter

\usepackage{babel}
\providecommand{\definitionname}{Definition}
\providecommand{\remarkname}{Remark}
\providecommand{\theoremname}{Theorem}

\newcommand{\bitstring}{{bitstring}}
\newcommand{\bitstrings}{{bitstrings}}

\newcommand{\graph}{\mathcal{G}}
\newcommand{\cC}{\mathcal{C}}
\newcommand{\cD}{\mathcal{D}}

\newcommand{\cN}{\mathcal{N}}
\newcommand{\cP}{\mathcal{P}}

\newcommand{\cS}{\mathcal{S}}
\newcommand{\cT}{\mathcal{T}}
\newcommand{\cV}{\mathcal{V}}

\newcommand{\Ex}{{\mathsf {Ex}}}
\newcommand{\dsf}{{\mathsf {d}}}
\newcommand{\nn}{{\mathbb {N}}}

\newcommand{\kk}{{\mathbb {K}}}

\newcommand{\bldc}{{\mbox{\boldmath $c$}}}

\newcommand{\bldw}{{\mbox{\boldmath $w$}}}

\newcommand{\bldalpha}{{\mbox{\boldmath $\alpha$}}}

\newcommand{\vertices}{\mathcal{V}}
\newcommand{\family}{\mathcal{F}}
\newcommand{\edges}{\mathcal{E}}

\title{\bf Forwarding Without Repeating: \\ Efficient Rumor Spreading in Bounded-Degree Graphs}
\author{
Vincent~Gripon, Vitaly~Skachek, and Michael~Rabbat\\
Department of Electrical and Computer Engineering\\
McGill University, Montr\'eal, Canada\\
Email: vincent-gripon@ens-cachan.org, vitaly.skachek@gmail.com, michael.rabbat@mcgill.ca
}
\date{}

\begin{document}

\maketitle
\thispagestyle{empty}

\begin{abstract}
We study a gossip protocol called \emph{forwarding without repeating} (\texttt{fwr}). The objective is to spread multiple rumors over a graph as efficiently as possible. \texttt{fwr} accomplishes this by having nodes record which messages they have forwarded to each neighbor, so that each message is forwarded at most once to each neighbor. We prove that \texttt{fwr} spreads a rumor over a strongly connected digraph, with high probability, in time which is within a constant factor of optimal for digraphs with bounded out-degree. Moreover, on digraphs with bounded out-degree and bounded number of rumors, the number of transmissions required by \texttt{fwr} is arbitrarily better than that of existing approaches. Specifically, \texttt{fwr} requires $O(n)$ messages on bounded-degree graphs with $n$ nodes, whereas classical forwarding and an approach based on network coding both require $\omega(n)$ messages. Our results are obtained using combinatorial and probabilistic arguments. Notably, they do not depend on expansion properties of the underlying graph, and consequently the message complexity of \texttt{fwr} is arbitrarily better than classical forwarding even on constant-degree expander graphs, as $n \rightarrow \infty$. In resource-constrained applications, where each transmission consumes battery power and bandwidth, our results suggest that using a small amount of memory at each node leads to a significant savings.
\end{abstract}

\newpage 
\setcounter{page}{1}

\section{Introduction}

We consider the now widely-studied problem of spreading rumors over a
graph. Initially a subset of nodes have a rumor and we are interested
in how long it takes to spread these messages to all nodes in the
graph. In this paper we are particularly interested in a variant of
the classical push-based forwarding algorithm. In the variant, which
we refer to as \emph{forwarding without repeating}, a node already
holding one rumor forwards it to a neighbor only one time. Thus, each
node must record which rumors it has sent to which neighbors. In
bounded degree networks, if the number of rumors is bounded, we show
that the number of transmissions required for forwarding without
repeating to spread the rumors to all nodes is infinitely
smaller than that of standard
(push-based) forwarding. Moreover, the delay of forwarding without
repeating is within a constant factor of optimal for any strongly
connected graph. 

Our motivation comes from the variety of resource-constrained applications, where reducing the number of transmissions is important for efficiency of the system. These include efficient broadcast of code updates in networks of embedded wireless sensing and/or control devices~\cite{Trickle,Mesbahi} and content dissemination in mobile social networks~\cite{Chaintreau}. In these applications, devices have limited onboard energy resources and limited bandwidth. Reducing the number of  transmissions can both free up communication resources and conserve energy resources. The forwarding without repeating algorithm considered here requires the use of additional memory at each node. However, in networks with bounded degree and where the goal is to spread a bounded number of rumors, we argue that using additional memory resources at each node may be beneficial if it leads to considerable savings in the number of transmissions.

The literature on rumor spreading is vast. Here we briefly review the developments which are most relevant to our discussion. The seminal work of Frieze and Grimmet~\cite{Frieze} introduces the problem of rumor spreading (a.k.a.~gossip) and studies push-based gossip protocols on complete graphs. Demers et al.~\cite{Demers} consider the use of gossip algorithms for lazy updating of distributed databases. They characterize the performance of push-, pull-, and push-pull algorithms, again on complete graphs. Chierichetti, Lattanzi, and Panconesi~\cite{Italian} provide the best-known upper bound on the number of iterations required for synchronous gossip (where each node communicates with one neighbor at each iteration) to spread a rumor over the entire network. They show that in a network of $n$ nodes with conductance $\phi$, the push-pull strategy reaches every node within $O(\frac{\log^2 \phi^{-1}}{\phi} \cdot \log n)$ rounds, with high probability. In addition, they show that in graphs which satisfy a regularity condition (all nodes have the same degree, to within a multiplicative factor), the push and pull strategies satisfy similar bounds on their own. Relatively little work has focused on bounding the number of transmissions required to spread rumors (with the exceptions of~\cite{Karp} and \cite{Elsasser}, which characterize the message complexity for complete graphs and Erdos-Renyi random graphs, respectively). Most work in this direction is of the form: if a network of $n$ nodes requires $O(r)$ rounds to spread a message, then the number of transmissions is no more than $O(n \cdot r)$.

In a slightly different direction, algebraic gossip protocols are proposed and studied in~\cite{Deb}. There the problem of spreading multiple messages is considered, and a network coding approach is adopted. Rumors are each finite-length bitstrings, and rather than gossiping on individual strings, the authors consider schemes where nodes form and transmit random linear combinations of the messages they have already received at each iteration. After every node has received sufficiently many of these coded strings, it can recover the individual rumors by solving a system of linear equations. In graphs with good expansion properties, this leads to a scheme where all nodes receive all messages in less time than it would take to spread all messages individually (without coding).

In this paper we compare three approaches to rumor spreading: classical (push-based) forwarding (\texttt{for}), the network coding approach (\texttt{nca}), and forwarding without repeating (\texttt{fwr}). We prove the following results:
\begin{itemize}
\item First, we analyze \texttt{fwr} in the synchronous time setting and we show that the number of rounds required for it to converge is within a constant factor of the diameter of the network, which is a trivial lower bound for rumor spreading on any graph.

\item Next, we analyze all three algorithms in the asynchronous time
  setting, where only a single pair of neighboring nodes exchange
  messages in each round. A trivial lower bound on the number of
  rounds required in this setting is $n$, the number of nodes in the
  network. For strongly connected (di)graphs with bounded out-degree,
  we show that \texttt{fwr} is within a constant factor of this lower
  bound. Moreover, when considering any family of graphs with bounded
  degree, we show that both \texttt{for} and \texttt{nca} require a
  number of rounds which is $\omega(n)$ (for $n \rightarrow
  \infty$). Since the number of transmissions is proportional
  to the number of rounds in the asynchronous setting, this implies
  that \texttt{fwr} is arbitrarily more efficient than both
  \texttt{for} and \texttt{nca} in large networks in terms of message
  complexity.

\end{itemize}
Our proofs use a combination of combinatorial and probabilistic arguments. Notably, our results do not involve assumptions about the conductance or other expansion-related properties of the underlying graph. Hence, these conclusions apply even to the family of constant-degree expander graphs~\cite{Reingold}. Another notable feature of this work is that, in contrast to most previous work which assumes the graph is undirected, our results are obtained in the general setting of strongly connected directed graphs.

\section{Notation, Assumptions, and Problem Formulation} \label{sec:notation}

Consider a connected directed graph $\graph = ( \vertices, \edges )$, where $\vertices$ is the finite set of vertices, $|\vertices | = n$, and $\edges \subseteq \{ (u, v) \; : \; u, v \in \vertices, u \neq v \}$ 
is the set of directed edges. Each vertex represents an
agent in a distributed system (e.g., terminals in a communication
network). Edges represent communication media (e.g., wires in an
Ethernet network, proximity in a wireless network).
We say that a vertex $v'$ is a \emph{successor} of the vertex $v$ if
$(v,v')\in\edges$. We also say that $v$ is a \emph{predecessor} of
$v'$. A neighbor of $v$ is a vertex that is either its predecessor or its successor.

We impose the \textit{gossip transmitter constraint}, which
implies that at any given time, any vertex in the network can only send a
{\bitstring} to \emph{at most one} of
its successors (see also the definition in~\cite{Mosk-Aoyama}). 
Conversely, a vertex may receive multiple {\bitstrings}
from different predecessors at the same time. 
This assumption is reasonable considering a wired network, for
instance, where multiple incoming messages can be buffered.

The out-degree $d^{out}(v)$ (respectively, in-degree $d^{in}(v)$) of a vertex $v$ is the number of
successors (respectively, predecessors) it has, namely, 
\[
d^{out}(v)=|\{v' \; : \; (v,v')\in\edges\}| \quad \text{ and } \quad d^{in}(v)=|\{v' \; : \; (v',v)\in\edges\}| \; . 
\]
The out-degree and in-degree of the graph $\graph$
are the maximum out-degree and the maximum
in-degree, respectively, taken over all vertices in $\graph$:
\[
d^{out}(\graph) = \max_{v \in \vertices} d^{out}(v) \quad \text{ and } \quad 
d^{in}(\graph) = \max_{v \in \vertices} d^{in}(v).
\]

A path in $\graph$ is a sequence of vertices $v_0,v_1,\dots,v_\ell$
such that for all $i = 0,1, \cdots, \ell-1$ holds $(v_i,v_{i+1})\in \edges$; $\ell$
is called the length of the path. A shortest path between $v$ and $v'$
is a path such that the first vertex is $v_0 = v$, the last is $v_\ell = v'$ and the
length of the path is minimum over all paths from $v$ to $v'$. We refer to the \emph{distance} between $v$ and $v'$ as the length of the shortest path. (Note that, in general, this distance is not a proper metric; it is not necessarily symmetric.) The diameter
$D(\graph)$ is defined as the maximum length of a shortest path
between any two vertices in $\graph$.
If there exists a path between all $v \in \cV$ and $v' \in \cV$,
we say that the graph is strongly
connected. If $n\geq 2$, a strongly connected graph is such that for all $v \in \cV$ holds $d^{out}(v)\geq 1$
and $d^{in}(v) \geq 1$. In the sequel,
the graph $\graph$ is always assumed to be strongly connected.

To perform an asymptotic analysis on the performance of the algorithms in the paper, 
we consider a family $\family = \{ \graph_i \}_{i=1}^\infty$ of
graphs with increasing number of vertices. We say that the family of graphs 
$\family$ has a bounded out-degree if
\[
\exists \dsf \in \nn \; \; \mbox{ such that } \; \; \forall \graph_i \in \family \; : \; d^{out}(\graph_i) \leq \dsf \; . 
\]
Similarly, the family $\family$ is said to have a bounded in-degree if 
\[
\exists \hat{\dsf} \in \nn \; \; \mbox{ such that } \; \; \forall \graph_i \in \family \; : \; d^{in}(\graph_i) \leq \hat{\dsf} \; . 
\]

Examples of families of bounded degree graphs include:
a family of directed symmetric chain graphs with $\dsf = \hat{\dsf} = 2$;
a family of $q$-ary rooted directed trees 
with $\dsf = q$ and $\hat{\dsf} = 1$; a family of
directed symmetric toroidal lattices with $\dsf = \hat{\dsf} = 4$; and a family of ring graphs 
with $\dsf = \hat{\dsf} = 2$.

%--------------------------------------------------------------------
%\medskip

Throughout this paper, we consider the following communication scenario: initially, each vertex has only the
knowledge of its successors and, possibly, its predecessors (the global structure of the graph is
unknown). The representation each vertex has of its neighbors is
relative, meaning that there is no global indexing of nodes in the
graph.\footnote{This setup is referred to as \emph{anonymous networking} 
in the literature, see, e.g.,~\cite{Attiya}.} Also, $m$ distinct {\bitstring}s are initially dispersed throughout the network (each {\bitstring} is initially held by one vertex, different {\bitstring}s may be held by different initial vertices).
The {\bitstring}s are assumed to be drawn from a very large set, such as $\kk^t$, where $\kk$ is a finite field 
(not necessarily binary) and $t$ is an integer.  

The goal of the \emph{covering problem} is to deliver all $m$ {\bitstring}s to all vertices in the graph. The transmission can be carried 
out either in {\bf synchronous} or in {\bf asynchronous} mode. In either setting, at most one
{\bitstring} is transmitted at a given time from a vertex to one of its
successors. Typically, at each point of time, a vertex can transmit to its successors 
at most one \bitstring. 

\begin{description}
\item[In the synchronous model,] 
the network is provided with an universal clock, such that 
all the transmissions proceed in synchronized rounds, and all vertices can transmit at each round. 
\item[In the asynchronous model,] each vertex is equipped with a clock that ticks according to a rate-1 Poisson point process. 
The clocks at different nodes are assumed to be independent of each other. When the clock ticks at a given vertex, that vertex can transmit (although it may also decide not to do so). Via the additive property of Poisson processes, the collection of clock ticks across the network can be viewed as being generated by a single rate-$n$ Poisson process, where each clock tick is assigned uniformly and independently to one of the vertices. Hence, the next clock tick is always uniform and \emph{i.i.d.} over the vertices in the network.
\end{description}

In both synchronous and asynchronous setups, we consider a \emph{call} model; i.e., where the transmitting vertex 
calls a communicating partner chosen from all its successors.
As it was observed in~\cite{Kempe}, the gossip algorithm typically consists of two decision layers. 
The first layer, \emph{a basic gossip algorithm}, is related to the decision made by a vertex about 
which of its successors to contact when its clock ticks. This decision can be 
made in either deterministic or randomized manner. For example, vertex can select one of its 
successors, at random using uniform distribution.  

The second layer, \emph{a gossip-based protocol}, is related to what content is transmitted to the successor and how the internal status of the transmitter and receiver is updated.
This content can be just one of the {\bitstring}s the transmitting vertex already has in its possession, 
or it can be a linear combination of some {\bitstring}s, akin the network coding (see~\cite{Deb}). 

%----------------------------------------------------------------

We consider three problems that are specific cases of what is usually
called \emph{gossip covering} or \emph{information spreading} problems. 
\emph{Randomized gossip algorithms} are usually employed to solve
these problems. The three problems are the following.
\begin{enumerate}
\item \textbf{Expected covering.} In expected covering, the
  objective is that each vertex in $\graph$ receives all {\bitstring}s
at least once in expectation after an infinite number of clock ticks.
\item \textbf{Almost sure covering.} In almost sure covering,
  the objective is that, after an infinite number of clock ticks, each vertex in $\graph$ has received all the {\bitstring}s with probability
one.
\item \textbf{Sure covering.} In sure covering, the objective is to guarantee (i.e., for \emph{every} sample paths) that each vertex
  in $\graph$ receives all the {\bitstring}s after a finite number of transmissions.
\end{enumerate}

There exists a hierarchy between those three problems as stated in the following theorem, which is proved in Appendix~\ref{app:proofOfThm1}.

\begin{theorem} 
Sure covering problem is strictly more difficult than almost sure covering problem that is
itself strictly more difficult than expected covering problem.
\label{thm:almostsureversusexpected}
\end{theorem}

By slightly abusing the language, we sometimes say the function $f(\cdot)$ is a bound for the covering
problem, if $f(\cdot)$ is a bound for all three covering problems.

\begin{comment}
Classically, there exists two different models for time in distributed
computation: synchronous and asynchronous. In synchronous time setting, the network
is provided with a universal clock such that at each time slot, all
vertices can send a {\bitstring} to any of its successors independently of the
choices of the other vertices. In asynchronous time setting, each vertex is given a
rate 1 Poisson timer before it can send a {\bitstring} to one of its
successors. Then, recursively and independently of whether it sent a {\bitstring} or not, a new
rate 1 Poisson timer is instantiated before the next possible
transmission of this vertex. As far as the order in which vertices
are able to communicate is concerned, this model is equivalent to the uniform
random selection of vertices in the network.
\end{comment}

We define two quantities characterizing the gossip algorithms: 
\emph{the delay} and \emph{the number of transmissions}.  

\begin{definition} In a synchronous gossip algorithm: the delay is the total number of 
transmission rounds.

In an asynchronous gossip algorithm: the delay is the total number of times any vertex
is allowed to transmit (i.e., the number of global clock ticks).
\end{definition}

\begin{definition}
The number of transmissions in a (synchronous or asynchronous) gossip
algorithm is the total number of actual transmissions used in the algorithm.
\end{definition}  

The number of
transmissions in the asynchronous time setting is always smaller or equal to the delay. 
It is also smaller or equal to $n$ times the delay of the same algorithm in synchronous time
setting. The delay characterizes the time required to solve one of the covering problems. 
The number of transmissions is a different measure, which is significant if the main objective is to save
non-necessary transmissions (for instance, the power consumption, given that each transmission consumes a 
constant amount of energy). 
Note that these quantities tend to infinity as $n$ grows. In that case we
are interested in comparing the asymptotic behavior of ratios of these quantities 
for different algorithms.

%There exist known solutions for the almost sure covering problem: among them are
%the classical random forwarding (\texttt{for}) and the network coding
%approach (\texttt{nca})~\cite{Deb}. In this paper, we introduce new lower bounds
%for almost sure covering using these algorithms. We study a distributed gossip
%algorithm, which we call forwarding-without-repeating (\texttt{fwr}) that
%solves the problem of sure covering. We
%show that this algorithm, (a) when used in the synchronous time
%setting and on a family of graphs with bounded out-degree
%has an expected time in a constant factor of the lower bound, (b) when
%used in the asynchronous time setting and on a family of graphs with
%bounded out-degree, has an almost sure covering number of
%transmissions in a constant factor of the lower bound. Moreover, in
%the latter scenario, we show that if the family of graphs also has bounded in-degree, then the ratio of the number of
%transmissions using \texttt{fwr} to that of either of the two other
%algorithms tends to zero; i.e., \texttt{fwr} uses infinitely fewer transmissions as $n \rightarrow \infty$.

\section{Algorithms} \label{sec:algorithms}

We consider three different algorithms: \texttt{for}, \texttt{nca} and
\texttt{fwr}. These algorithms are executed at each vertex in the process  
of {\bitstring} transmission and receiving. 
We denote by $\cS(v)$, $\cP(v)$ and $\cC(v)$, respectively,
the set of successors, the
set of predecessors, and the set of {\bitstring}s known to a vertex $v$ at some time during execution of the 
algorithm.

\subsection{Classical forwarding (\texttt{for})}

The \texttt{for} algorithm is straightforward. Each vertex $v$ 
chooses one of its successors randomly and uniformly from $\cS(v)$,
and also chooses one {\bitstring} randomly and uniformly from $\cC(v)$, and sends the chosen 
{\bitstring} to the chosen successor. If it owns no
\bitstring, then it performs no transmission. Note that in some
literature this algorithm is referred to as the \emph{push} version
of \texttt{for}. A pseudo code version of \texttt{for} is presented in
Algorithm \ref{alg:for}. 

\begin{algorithm}[h]
\small
 \SetAlgoLined
 Sending a {\bitstring}:\\
 \Begin{
   \If{$\cC(v)\neq\emptyset$}{
     select $v'$ at random in $\cS(v)$\;
     select $\bldc$ at random in $\cC(v)$\;
     send $\bldc$ to $v'$\;
}}
$\,$ \\
 Receiving a {\bitstring} $\bldc$:\\
   \Begin{
   $\cC(v)\leftarrow \cC(v)\cup \{\bldc\}$\;
}
$\,$ \\
 \caption{Classical forwarding (\texttt{for}) algorithm for vertex
   $v$.}
\label{alg:for}
\end{algorithm}

\subsection{The network coding approach (\texttt{nca})}

In the network coding approach, the {\bitstring}s are assumed to be 
vectors over the finite field $\mathbb{K}$. In each transmission step, 
the vertex sends a random linear combination of the {\bitstring}s it 
already has in its possession. The coefficients of the vectors in this 
linear combination are selected randomly and uniformly from $\mathbb{K}$. 
Denote by $\bldc_1, \bldc_2, \cdots, \bldc_m$
the original {\bitstring}s.

In this approach, the receiving vertex has to be able to recover the 
original {\bitstring}s from the received linear combinations. To do
so, it must obtain $m$ linearly independent vectors. 
Moreover, the coefficients of the encoding must also be known to 
the receiver. Thus, the transmitted message contains both 
a header with the coefficients and the information payload. 
In other words, the message is a pair $(\bldc,\textbf{w})$ where 
$\textbf{w} = (w_1, w_2, \cdots, w_m)$ is a vector of $m$ elements in $\mathbb{K}$, 
and $\bldc = \sum_{i=1}^m w_i \bldc_i$.

The receiving vertex stores in its memory the list of pairs received so far, 
\[
\mathcal{R}(v)= \big\{ \; \left(\hat{\bldc}_i, \hat{\bldw}_i \right) \; \big\}_{i=1}^r  \;, 
\] 
where $r$ is the current size of the list, and $\hat{\bldw}_i = (\hat{w}_{i,1}, \hat{w}_{i,2}, \cdots, \hat{w}_{i,m})$.  
For convenience, define the $r \times m$ matrix $W = \left( \hat{w}_{i,j} \right)_{i = 1,2, \cdots, r; \; j = 1, 2, \cdots, m}$. 

As the size of $\mathcal{R}(v)$ grows, the
vertex may try to solve the corresponding system of linear equations in order 
to retrieve the original {\bitstring}s. Pseudo code for \texttt{nca} is presented in Algorithm~\ref{alg:nca}.

\begin{algorithm}[h]
\small
\SetAlgoLined
Sending a \bitstring:\\
\Begin{
draw a vector $\bldalpha = (\alpha_1, \alpha_2, \cdots, \alpha_r)$ uniformly at random in
$\mathbb{K}^{r}$\;
select a successor $v'$ at random in $\cS(v)$\;
send $\big( \bldc = \sum_{i=1}^r \alpha_i \hat{\bldc}_i, \; \bldw = \bldalpha W  \big)$ to $v'$\;
}
$\,$ \\
Receiving a {\bitstring} $(\bldc,\textbf{w})$:\\
\Begin{
add $(\bldc,\textbf{w})$ to $\mathcal{R}(v)$\;
$r \leftarrow r + 1$ \;
check for a full-rank $m \times m$ submatrix $Z$ of $W$ \;
\If{$Z$ is invertible}{
inverse $Z$ and find the initial {\bitstring}s\;
}
}
$\,$ \\
\caption{Network coding approach (\texttt{nca}) algorithm for node
  $v$.}
\label{alg:nca}
\end{algorithm}

Lemma \ref{lemma:ncafor} gives a relation for the delays and
number of transmissions between \texttt{nca} and \texttt{for}:

\begin{lemma}
For any graph, the delay and the number of
transmissions required by \texttt{nca} are at least that required by
\texttt{for} for the same objective (expected or almost sure covering)
with a single \bitstring.
\label{lemma:ncafor}
\end{lemma}

\begin{proof}
Let us consider \texttt{nca} on a graph with $m$ {\bitstring}s. 
In this case, typically the vertex sends out a linear combination 
of the received messages. 
It is straightforward to see 
that the propagation of messages in \texttt{nca} is analogous to the
propagation of an initial {\bitstring} in \texttt{for}.

Since it is clear that a vertex that receives no message cannot retrieve the
initial {\bitstring}s, the delay and the number of transmissions needed in \texttt{nca} is
greater than or equal to that of \texttt{for}.
\end{proof}

Note that in both \texttt{for} and \texttt{nca} 
the selection of the receiving vertex is fully random.
Therefore, these algorithms cannot solve the sure covering problem for
most of non-trivial graphs.

\subsection{Forwarding without repeating (\texttt{fwr})}

When taking a closer look at \texttt{for}, it appears that time and 
transmissions are wasted when a vertex receives a message it already has. This situation, in particular, 
happens when the same predecessor sends the same
message to the same successor twice. As we show below, a simple amendment to \texttt{for}, which prevents
repeated transmissions of the same {\bitstring}, 
can lead to a meaningful improvement.

In the \texttt{fwr} algorithm, each vertex stores in
its memory the list of pairs $\mathcal{L}(v) 
\subseteq \cC(v) \times \cS(v)$ of {\bitstring}s and successors it has sent
messages to, including which specific {\bitstring}s were sent, and it also records the predecessors from which it has received messages. By using this data, the vertex will not send a {\bitstring} $\bldc$ to a
successor $v'$ if the pair $(\bldc,v')$ already appears in
$\mathcal{L}(v)$. Thus some redundant transmissions may be avoided.
Note that we slightly abuse the notation by treating $\mathcal{L}(v)$ as a set. 

Pseudo code version for \texttt{fwr} is given in Algorithm \ref{alg:fwr}. Note that for the sake of simplicity, Algorithm~\ref{alg:fwr} is not presented in a way consistent
with the separation in two distinct layers (basic gossip algorithm and gossip-based protocol).
However, it is straightforward to change the presentation to make it consistent with that model.

\begin{algorithm}[h]
\small
\SetAlgoLined
Sending a {\bitstring}:\\
\Begin{
\If{ $\left( \cC(v)\times \cS(v) \right) \backslash \mathcal{L}(v) \neq \emptyset$}{
select a pair $(\bldc,v')$ uniformly at random in $\left( \cC(v)\times \cS(v)\right) \backslash \mathcal{L}(v)$\;
send the {\bitstring} $\bldc$ to $v'$\;
add $(\bldc,v')$ to $\mathcal{L}(v)$\;
}
}
$\,$ \\
Receiving a {\bitstring} $\bldc$ from $v'$:\\
\Begin{
add $(\bldc,v')$ to $\mathcal{L}(v)$\;
}
$\,$ \\
\caption{Forwarding without repeating (\texttt{fwr}) algorithm for
  node $v$.}
\label{alg:fwr}
\end{algorithm}

When a vertex receives a {\bitstring} $c$ from predecessor $v'$, it stores the pair $(\bldc, v')$
in the list $\mathcal{L}(v)$. However, in the analysis in this paper, for the sake of simplicity 
we assume a slightly weaker version of 
the algorithm \texttt{fwr}. More specifically, we assume that upon \emph{receiving} a message, 
the vertex performs no update of the list $\mathcal{L}(v)$.

Since \texttt{fwr} can be viewed as an improved version of \texttt{for}, where some
redundant transmissions are avoided, it is immediate that the
delay and the number of transmissions required using \texttt{fwr} is
less or equal to those required when using \texttt{for}. 

In the following sections, we derive some bounds for covering
problems for both synchronous and asynchronous settings.

\section{Synchronous time setting} \label{sec:synchronous}

First, we derive a lower bound on the covering delay for any algorithm and any $m$ under the synchronous update model. 
Assume that only one {\bitstring} is used in the network. The number of vertices
that have obtained the {\bitstring} cannot grow by more than a factor of two during one step of
the algorithm, which corresponds to the case where all vertices having
the {\bitstring} send it to distinct vertices that do not have it yet. It follows that the delay for solving the covering problem using any algorithm on $\graph$, $\cD(\graph)$, is at least:

\begin{equation}
\cD (\graph) \ge \lceil \log_2(n) \rceil \; . 
\label{eq:ssyncmin}
\end{equation}

This lower bound is not tight for all graphs. An
alternative lower bound depends on the {\bitstring} diameter of a graph. The {\bitstring}
diameter $D_c(\graph)$ of the graph is the maximum length of a
shortest path between a vertex that initially has a {\bitstring} and any
other vertex.

\begin{theorem}
In the synchronous time setting, a lower bound for covering delay is:
\[
\cD (\graph) \ge D_c(\graph) \; . 
\]
\label{thm:syncmin}
\end{theorem} 
The proof of Theorem~\ref{thm:syncmin} appears in Appendix~\ref{app:syncmin}. 
\medskip

By using \texttt{fwr},
it is immediate that in at most $m \cdot d^{out}(v)$ steps (corresponding to the
case where $v$ has all $m$ {\bitstring}s), all the successors of $v$
will obtain the {\bitstring}s known to $v$. By using this principle recursively, one
can derive the following upper bound on the covering delay:

\begin{lemma}
When using \texttt{fwr} on a graph $\graph$ in synchronous settings, 
the covering delay is at most:

\begin{equation}
\cD(\graph) \le m \cdot d^{out}(\graph) \cdot D_c(\graph) \; . 
\label{eq:smaxsyncfwr}
\end{equation}
\end{lemma}

It follows from \eqref{eq:ssyncmin} and Theorem \ref{thm:syncmin} that
the delay in \texttt{fwr} is at a constant factor $m \cdot d^{out}(\graph)$ from the optimal number of steps,
for any family of graph with bounded out-degree and bounded number of
initial {\bitstring}s.

%-----------------------------------------------------------------------------

For specific graphs we can provide tighter bound on the sure covering time, as
we state in the following theorem, which is proved in Appendix~\ref{app:tree}.

\begin{theorem}
Let $\graph$ be a directed \emph{tree} with root $v_0$, and assume that $v_0$ has the {\bitstring} $\bldc$. 
Consider synchronous gossip algorithm \texttt{fwr} applied to this graph. 
Then, the number of transmission rounds, $\cN(\graph)$, needed for sure covering of $\graph$ is 
upper bounded by  
\[
\cN(\graph) \le \max_{\Phi} \left\{ \sum_{v \in \Phi} d^{out}(v) \right\} \; ,  
\]
where $\Phi$ runs over all paths from $v_0$ to the leaves of $\graph$, and where
$d^{out}(v)$ denotes the out-degree of $v$ in the tree $\graph$. 
\label{thrm:tree}
\end{theorem}

\section{Asynchronous time setting} \label{sec:asynchronous}

Similar to the scenario of the synchronous time setting, we can
derive a trivial lower bound on the covering number of
transmissions in asynchronous mode on graph $\graph$. Moreover, since the number of transmissions is smaller than the delay, 
this bound also holds for the covering delay. This lower bound
is
\begin{equation}
\cN(\graph) \ge n - 1 \; , 
\label{eq:lowerbound}
\end{equation}
and is a direct consequence of the fact at most one new vertex receives
the {\bitstring} at each asynchronous step (and $n-1$ vertices initially do not have one
of the {\bitstring}s).

As far as the number of transmissions is concerned, we can derive some upper bounds for \texttt{fwr}. Those are
presented in Theorem \ref{thm:sasyncmaxfwr} (which is proved in Appendix~\ref{app:sasynchmaxfwr}) and Corollary \ref{cor:sasyncmaxfwr}.

\begin{theorem}
The number of transmissions required to solve
the covering problem using \texttt{fwr} on $\graph$ is
\[
\cN(\graph) \leq m |\edges| \; .
\]
\label{thm:sasyncmaxfwr}
\end{theorem}

\begin{corollary}
As a direct consequence of Theorem \ref{thm:sasyncmaxfwr}, we have:
\[
\cN(\graph) \leq m n \cdot d^{out}(\graph) \;  .
\]
\label{cor:sasyncmaxfwr}
\end{corollary}

It follows from \eqref{eq:lowerbound} and Corollary \ref{cor:sasyncmaxfwr} that
\texttt{fwr} is asymptotically optimal (up to a multiplicative constant)
in the number of required transmissions
when applied to a family of graphs with bounded out
degree and with a bounded number of initial {\bitstring}s.

%\subsection{Number of transmissions for algorithm \texttt{for} on general graphs}
\medskip 
Let us now focus on estimating bounds for the number of transmissions required when
using \texttt{for} for any family of graph with bounded in-degree and
bounded number of initial {\bitstring}s. (See Appendix~\ref{app:chain} for an analysis tailored specifically to the family of chain graphs.)

\begin{theorem}
Assume that \texttt{for} is used in asynchronous time setting on the graph $\graph$. 
Then, there exists vertex $v \in \cV$, such that the probability that $v$ received all 
{\bitstring}s after at most $s$ transmissions, $P^*_{s}(v)$, satisfies:
\[ 
P^*_{s}(v) \leq 1-
\exp\left(-\frac{2 s \cdot d^{out}(\graph) \cdot d^{in}(v)}{n}\right) \;  . 
\]
\label{thm:sasyncfor}
\end{theorem}

\begin{proof}
First, note that the expression does not depend on the number
of initial {\bitstring}s. It is obvious that a larger
number of {\bitstring}s requires more transmissions
using \texttt{for} and thus we focus on the simple case where $m=1$.
Let $\bldc$ be the corresponding {\bitstring}. 

It will be convenient to say that a vertex of a graph is \emph{contaminated} if it
has obtained the {\bitstring} $\bldc$. After $s$ transmissions, the set of
contaminated vertices in the graph $\graph$ is denoted $C_s$.
Assume that $\graph$ contains at least 2 vertices, and take some $k \in \nn$ such that $n \geq k \geq 1$. 
Run \texttt{for} until the point when  
$|C_s|=k$ for the first time, and let $s$ be the corresponding number of transmissions 
($s$ is well defined since the size of $C_s$ increases by at most one after
each transmission).
The probability
that the next vertex to receive the {\bitstring} at the iteration $s+1$ is a given vertex $v$ can be
upper bounded. As a matter of fact, the algorithm will use any edge at
random among those starting at vertices in $C_s$. We note that,
since $|C_s|=k$ and since the graph is strongly connected, the number
of such edges is at least $k$ (recall that for each $v\in C_s$, we have $d^{out}(v) \geq 1$). 
At most $d^{in}(v)$ of these edges are
connected to a given vertex $v \not\in C_s$. Note that in
contrast to the vertices that are drawn uniformly in the asynchronous
time setting version of \texttt{for}, the probability to use one edge is not
uniform if the degrees of vertices are not equal. However, the ratio between the probability to
choose the most probable edge and the less probable edge is at most
$d^{out}(\graph)$. Therefore, the conditional probability
that the vertex $v$ will receive the {\bitstring} at the transmission $s+1$ is
\begin{eqnarray*}
P_{s+1}(v) &\leq& \frac{\big| \text{edges connecting a vertex in $C_s$ to $v$} \big| \cdot d^{out}(\graph)}
{\big| \text{edges starting at vertices in $C_s$} \big|}\\
&\leq& \frac{d^{out}(\graph) \cdot d^{in}(v)}{k}\;.
\end{eqnarray*}
(This probability is conditioned on the assumption that $|C_{s}| = k$, and that $v \notin C_s$, as we mentioned above.)

Now, run \texttt{for} until the point when  
$|C_s| \geq n/2$ for the first time, and let $s_0$ be the corresponding number of transmissions. 
Select a vertex $v$ that was not contaminated in the course of the
first $s_0$ transmissions (which is well defined since $n\geq 2$). The conditional 
probability (conditioned on the assumptions $|C_{s_0}| \geq n/2$ and $v \notin C_s$) that $v$ is contaminated after
$s_0 + 1$ transmissions is therefore:
\begin{equation}
P_{s_0+1} (v) \leq \frac{2 \cdot d^{out}(\graph) \cdot d^{in}(v)}{n}\;.
\label{eq:pv}
\end{equation}

Since all transmissions are independent of each other, the probability that $v$ is contaminated after 
$s=s_0 + s'$ transmissions is:
\begin{eqnarray*}
P^*_{s}(v) &\leq& 1-\left(1-\frac{2 \cdot d^{out}(\graph) \cdot d^{in}(v)}{n}\right)^{s'} \nonumber \\
&\leq & 1- \exp\left(-\frac{2s' \cdot d^{out}(\graph) \cdot d^{in}(v)}{n}\right)\;.
\end{eqnarray*}
Since $s'\leq s$, this completes the proof.
\end{proof}

%\medskip
From Theorem~\ref{thm:sasyncfor} we arrive at the following conclusion.

\begin{corollary}
Suppose that the algorithm \texttt{for} is applied in the asynchronous setting to 
a family of graphs with bounded in- and out-degree. Then the number of
transmissions $s$ for almost sure covering of the graph satisfies:
\begin{equation}
s=\omega(n)\;.
\end{equation}
\label{cor:asyncfor}
\end{corollary}

\begin{proof}
Observe that almost sure covering implies that every vertex $v \in \graph_i$ 
receives all the $m$ messages with probability one. 
By picking a vertex $v$ and applying analysis as in the proof of Theorem~\ref{thm:sasyncfor}, we
obtain that 
\[
\frac{2 s \cdot d^{out}(\graph) \cdot d^{in}(v)}{n}\to \infty \; .
\]

Denote by $\family = \{ \graph_i \}_{i=1}^\infty$ the family of
graphs of bounded degree. Then, 
\[
\exists \hat{\dsf} \in \nn \quad \mbox{ such that } \quad
\forall i, \; \forall v\in \graph_i \; : \; d^{in}(v) \leq
d^{in}(\graph_i) \leq \hat{\dsf}
\]
and 
\[
\exists {\dsf} \in \nn \quad \mbox{ such that } \quad
\forall i \; : \; d^{out}(\graph_i) \leq {\dsf} \; , 
\] 
and thus we obtain 
\begin{equation}
\dsf \cdot \hat{\dsf} \cdot \frac{s}{n} \to \infty \; .
\label{infty}
\end{equation}
If $\lim_{n \to \infty} \frac{s}{bn} \le 1$, for some constant $b \ge 0$, 
condition~(\ref{infty}) does not hold. Therefore, we conclude that 
$s$ grows faster than $bn$ for any $b \ge 0$, i.e. $s = \omega(n)$.    

\end{proof}

Using Corollaries \ref{cor:sasyncmaxfwr}, \ref{cor:asyncfor} and Lemma
\ref{lemma:ncafor}, we
conclude that the ratio of the number of transmissions needed to assure almost
sure covering using \texttt{fwr} to that using \texttt{for} or that
using \texttt{nca} tends to
zero when considering a family of graphs with bounded in- and out-degree and
bounded number of {\bitstring}s. Table \ref{table:async} below summarizes the main results of this extended abstract.

\begin{table}
\begin{center}
\begin{tabular}{|c|c|c|c|}
\hline
&\texttt{for}&\texttt{nca}&\texttt{fwr}\\
\hline
Expected covering & $\omega(n)$ & $\omega(n)$ &$\Theta(n)$\\
\hline
Almost sure covering & $\omega(n)$ & $\omega(n)$&$\Theta(n)$\\
\hline
Sure covering & N/A & N/A & $\Theta(n)$\\
\hline
\end{tabular}
\end{center}
\caption{Comparison of the number of transmissions for the
  different algorithms and different problems when considering
  asynchronous time setting and a family of graphs with bounded in- and
  out-degree and bounded number of initial {\bitstring}s.}
\label{table:async}
\end{table}

\appendix
\section{Proof of Theorem~\ref{thm:almostsureversusexpected}} \label{app:proofOfThm1}

First, we prove the loose inclusions:
\begin{itemize}
\item It is clear that sure covering is more difficult than almost
sure covering since if all vertices have all {\bitstring}s, they also
have all of them
with probability one. 
\item Suppose that \texttt{alg} is an algorithm that solves
almost sure covering.
After using \texttt{alg} on a graph, let us focus on a {\bitstring} $\bldc \in \kk^t$ and a
vertex $v \in \cV$. Denote by $\cT_\bldc(v)$ the random variable that
counts the number of times $v$ received the {\bitstring} $\bldc$ during the execution of \texttt{alg}
for infinite number of steps. The
probability that $v$ received $\bldc$ at least
one time is:
\[
P(\cT_\bldc(v) \ge 1) = \sum_{k=1}^\infty {P(\cT_\bldc(v) = k)} = 1 \;.
\]

On the other hand, the expected number of times $v$ received $\bldc$ is:
\begin{eqnarray*}
\Ex[\cT_\bldc(v)] &=& \sum_{k=1}^\infty{k \cdot P(\cT_\bldc(v) = k)} \nonumber \\
&=& \sum_{k=1}^\infty {P(\cT_\bldc(v) = k)} + \sum_{k=2}^\infty (k-1) \cdot P(\cT_\bldc(v) = k) \\
& \ge & \sum_{k=1}^\infty {P(\cT_\bldc(v) = k)} \\
& = & 1 \; . 
\label{eq:expected-covering}
\end{eqnarray*}
Thus \texttt{alg} always solves expected covering.
\end{itemize}

Next, let us prove the strict part of the theorem statement:
\begin{itemize}
\item Consider a graph that consists of two vertices $v$ and $v'$,
  and such that $\edges = \{ (v,v'), (v',v) \}$. Suppose that $v$ initially has 
  some {\bitstring} $\bldc$. We use the following algorithm: when $v$ can
  transmit (which happens infinitely often under both the synchronous and asynchronous model), it chooses
  with constant probability $0 < \alpha <1$ to send its {\bitstring} $\bldc$ to
  $v'$. This algorithm trivially solves almost sure covering as the
  probability that $v'$ receives the {\bitstring} $\bldc$ after $s$ steps is:
\[P_{v'}(s) = 1-(1-\alpha)^{s} \to 1\;.\]

On the other hand, it may happen with probability 0 that at each step,
$v$ does not communicate with $v'$ and thus this algorithm does not
solve the sure covering problem.
\item Using the same graph, let us now consider the following
  algorithm: the first time $v$ transmits, it draws a random
  variable with probability 1/2 to be one and probability 1/2 to be
  zero. If the obtained value is one, then $v$ sends its {\bitstring} to $v'$ exactly two times in a row; 
  if it is zero, then it never sends
  its {\bitstring} to $v'$. Trivially, the probability that $v'$ receives the {\bitstring}
  is 1/2, which means that this algorithm does not solve the almost sure
  covering problem. On the other hand, the expected number of times $v'$
  receives the {\bitstring} is $2 \cdot 1/2 = 1$, which means that this algorithm
  solves the expected covering problem.
\end{itemize}

\section{Proof of Theorem~\ref{thm:syncmin}} \label{app:syncmin}

Let $v \in \cV$ be a vertex that initially has
some {\bitstring} $\bldc$ and let $v' \in \cV$ be another vertex, such that there
exists a shortest path starting at $v$ and ending at $v'$ of length
$D_c(\graph)$. We use the fact that if the probability that $v'$ received
$\bldc$ is zero, then the expected number of times $v'$ received $\bldc$ is
also zero. We now proceed by induction:
\begin{itemize}
\item At the beginning, the only vertex in the graph that
  has a nonzero probability to have $\bldc$ is $v$.
\item After step $s$, only the vertices that
  are at a distance less than or equal to $s$ from $v$ have a nonzero probability
  of having received $\bldc$. At step $s+1$, the gossip transmitter 
  constraint implies that only the successors of
  vertices that have $\bldc$ may have a nonzero probability to receive~$\bldc$.
\end{itemize}

\section{Proof of Theorem~\ref{thrm:tree}} \label{app:tree}

 First, observe that the graph $\graph$ is covered if all its leaves have obtained 
the {\bitstring} $\bldc$. This is due to the fact that every vertex in the tree $\graph$ is located on the 
path from the root $v_0$ to some leaf. 

Second, let us estimate the number of transmission rounds required to deliver $\bldc$ to a particular leaf $v$. 
Let $\Phi$ be a path from $v_0$ to $v$, where $\Phi = (v_0, v_1, v_2, \cdots, v_\ell = v)$. 
In the worst case scenario, $v_0$ will send $\bldc$ to all its other successors before it sends $\bldc$ to $v_1$. 
Then, in the words case scenario, $v_1$ sends $\bldc$ to all its other successors before it sends it to $v_2$, and so on. 
In total, {\bitstring} $\bldc$ will be sent to $\sum_{v' \in \Phi} d^{out}(v')$ vertices when it reaches $v$. 

The final expression is obtained by taking the leaf $v$ which
maximizes the required expression. 

\section{Proof of Theorem~\ref{thm:sasyncmaxfwr}} \label{app:sasynchmaxfwr}

To obtain this result, we first observe that the quantities (for all $v \in \cV$) $m |\cS(v)| - |\mathcal{L}(v)|$ are 
non-increasing with the number of transmissions. More precisely, at each transmission, one of the quantities
$m |\cS(v)| - |\mathcal{L}(v)|$ decreases by one. Therefore, an upper bound for
the number of steps before all $\mathcal{L}(v)$'s are empty is:
\begin{eqnarray*}
\cN(\graph) &\leq& \sum_{v\in\vertices}{m |\cS(v)| - |\mathcal{L}(v)|}\\
&\leq& m \sum_{v\in\vertices}{|\cS(v)|}\\
&\leq & m |\edges|
\end{eqnarray*}

\section{Analysis of \texttt{for} on the chain graph} \label{app:chain}

\begin{theorem}
Solving expected covering problem using \texttt{for} on a family of graphs
with bounded out-degree in the asynchronous time setting requires
$\omega(n)$ transmissions.
\end{theorem}

\begin{proof}
Take the symmetric directed chain (di)graph depicted in Figure \ref{fig:linegraph} that
consists of $n+1$ vertices. This graph is defined by $\edges= \{(v_i, v_{i+1})\}_{i=0}^{n-1}\cup
\{(v_{i+1},v_{i})\}_{i=0}^{n-1}$. If the vertex $v_0$ owns the {\bitstring} $\bldc$, it takes 
at least $n$ transmissions to deliver this {\bitstring} to $v_n$. 

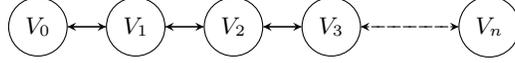
\begin{figure}
\begin{center}
\begin{tikzpicture}
\tikzstyle{every node}=[draw,circle];
\foreach \name/\position in {0/0,1/1.3,2/2.6,3/3.9}{
\node(\name) at (\position,0) {\footnotesize{$V_\name$}};
}
\node(n) at (6,0) {\footnotesize{$V_n$}};
\path[style=-stealth]
(0) edge (1)
(1) edge (2)
edge (0)
(2) edge (1)
edge (3)
(3) edge (2)
edge[dashed] (n)
(n) edge[dashed] (3);
\end{tikzpicture}
\end{center}
\caption{Chain graph where vertex $v_i$ is connected to vertices
  $v_{i-1}$ and $v_{i+1}$.}
\label{fig:linegraph}
\end{figure}

Assume that the vertex $v_1$ has the only {\bitstring} $\bldc$, and so $m = 1$. 
Note that the number of vertices having the
{\bitstring} $\bldc$ at the end of step $s$, which we denote by $k_s$, grows at most by one at
each transmission. Let $s_0$ be the smallest number of
transmissions such that $k_{s_0}= \lceil n/2 \rceil$. 

Since the set of vertices possessing the {\bitstring} is connected, the set of vertices having the {\bitstring} is 
\[
\displaystyle{\bigcup_{i=1}^{\lceil n/2 \rceil}{v_i}}.
\]
This set might also contain $v_0$, but then we ignore $v_0$ for the sake of simplicity 
of the analysis.

Denote by $E_i(s)$ the event that $v_i$ receives the {\bitstring} $\bldc$ at step $s$. 
We also denote by $P[E_i(s)]$ the conditional probability of the event $E_i(s)$ given that
$k_{s_0}= \lceil n/2 \rceil$ as above and $s \ge s_0$. 
 
We obtain that, for $s\geq s_0$ and $i=\lceil n/2 \rceil +1$,
\[ 
P[E_i(s)] \leq \frac{1}{k_{s_0}}\cdot \frac{1}{2} \; , 
\]
since $v_{i-1}$ is the next vertex
to transmit with probability $1/k_{s_0}$, and $v_{i-1}$ sends the
{\bitstring} to $v_i$ with probability $1/2$. Since $k_{s_0}= \lceil n/2 \rceil$, this implies that
\[
P[E_i(s)] \leq \frac{1}{n}\;.
\]

Since the transmissions are independent, the probability that $v_i$ receives the {\bitstring} in the first $s$
transmissions given that it did not receive it in the first $s_0$ transmissions, $s_0 \le s$, is:
\[ 
P \left[\bigcup_{s'=s_0+1}^{s}{E_i(s)} \right] \leq 1 - \left( 1-
\frac{1}{n}\right)^{s - s_0} \; .
\]
Take $s = b \cdot n + s_0$ for some constant $b \ge 0$ independent of $n$. 
It is obvious that $P[\bigcup_{s'=s_0+1}^{s}{E_i(s)}]$ is an increasing
function of $s$, and therefore, for any $s \leq b \cdot n + s_0$, we have:
\[
P \left[ \bigcup_{s'=s_0+1}^{s}{E_i(s)} \right] \leq  1 - \left( 1-
\frac{1}{n}\right)^{bn}\;.
\]

Since the events $E_{i-1}(s)$, $E_{i}(s)$ and $E_{i+1}(s)$ form a Markov chain, 
we can now upper-bound $P[E_{i+1}(s)]$:
\[
P[E_{i+1}(s)] \leq \left[1 - \left( 1-
\frac{1}{n}\right)^{bn}\right] \frac{1}{n}\;.
\]
Similarly, this yields:
\[
P \left[ \bigcup_{s'=s_0+1}^{s}{E_{i+1}(s)} \right] \leq \left[1 - \left( 1-
\frac{1}{n}\right)^{bn}\right]^2\;.
\]
This leads inductively to the formula:
\[
P \left[ \bigcup_{s'=s_0+1}^{s}{E_{n-1}(s)} \right] \leq {\left[1 - \left( 1-
\frac{1}{n}\right)^{bn}\right]^{n-1-i}}\;.
\]

Now, observe that  
\[
{\left[1 - \left( 1-
\frac{1}{n}\right)^{bn}\right]^{n-1-i}} \sim e^{-(n/2) e^{-b}} \; .
\]

Assume that the algorithm runs for $bn + s_0$ transmissions, where $b>0$ is a constant. 
To upper bound the expected number of times the {\bitstring} $\bldc$ is sent to $v_n$
during the execution of the algorithm, we note that the number of times $v_{n-1}$ 
transmits $\bldc$ to $v_{n}$ is less than $bn$. Then, the 
expected number of times that $\bldc$ is sent to $v_{n}$ is:
\[ 
\le bn \cdot e^{-(n/2) \cdot e^{-b}}\;.
\]
This value vanishes exponentially fast for sufficiently large values of $n$.
We conclude that the expected number of times that $\bldc$ is sent to $v_{n}$ is very close to zero. 
Therefore, if the number of transmissions in the algorithm is linear in $n$, the algorithm will not achieve 
an expected covering of the chain graph.
\end{proof}


\begin{thebibliography}{99}

\bibitem{Attiya}
H. Attiya, M. Snir, and M.K. Warmuth, ``Computing on an anonymous ring'', {\em Journal of the ACM}, vol.~35, no.~4, October~1988, pp.~845--875.
 
\bibitem{ChenAvin}
C. Avin and G. Ercal, 
``On the cover time and mixing time of random geometric graphs'', 
{\em Theoretical Computer Science,} vol. 380, Issue 1-2, July~2007, pp. 2--22.  

\bibitem{Censor-Hillel}
K. Censor-Hillel and H. Shachnai, ``Fast information spreading in graphs with large weak conductance'', 
{\em Proceedings of the Twenty-Second Annual ACM-SIAM Symposium on Discrete Algorithms (SODA'11),} 
San Francisco, California, Jan. 2011, pp. 440--448. 

\bibitem{Italian}
F. Chierichetti, S. Lattanzi, and A. Panconesi, 
``Rumour spreading and graph conductance'', 
{\em Proceedings of the Twenty-First Annual ACM-SIAM Symposium on Discrete Algorithms (SODA'10),} 
Austin, Texas, Jan. 2010, pp. 1657--1663. 

\bibitem{Deb} 
S. Deb, M. M\'edard, and C. Choute, 
``Algebraic Gossip: A Network Coding Approach to Optimal Multiple Rumor Mongering'',
{\em IEEE Trans. Inform. Theory,\/} vol.~52, no.~6, June 2006, pp. 2486--2507. 

\bibitem{Demers}
A.J.~Demers, D.H.~Greene, C.~Hauser, W.~Irish, J.~Larson, S.~Shenker, H.E.~Sturgis, D.C.~Swinehart, and D.B. Terry, ``Epidemic algorithms for replicated database maintenance'', {\em Proceedings of the Sixth Annual ACM Symposium on Principles of Distributed Computing (PODC'87),} 1987, pp. 1--12.

\bibitem{Elsasser}
R.~Els\"{a}sser, ``On the communication complexity of randomized broadcasting in random-like graphs'', {\em Proceedings of the 18th Annual ACM Symposium on Parallel Algorithms and Architectures (SPAA'06),} 2006, pp.~148--157.

\bibitem{Frieze}
A.~Frieze and G.~Grimmett, ``The shortest-path problem for graphs with random arc-lengths'', {\em Discrete Applied Mathematics,} vol.~10, 1985, pp.~57--77.

%\bibitem{Tsi}
%J. M. Hendrickx, A. Olshevsky, and J. N. Tsitsiklis, 
%``Distributed Anonymous Discrete Function Computation and Averaging'', 
%{\em IEEE Transactions on Automatic Control,} vol. 56, no.~10, Oct. 2011, pp. 2276--2289.

\bibitem{Chaintreau}
S.~Ioannidis and A.~Chaintreau, ``On the strength of weak ties in mobile social networks'', {\em Proceedings of the ACM Workshop on Social Network Systems (SNS)}, March 2009.

\bibitem{Karp}
R.~Karp, C.~Schindelhauer, S.~Shenker, and B.~V\"{o}cling, ``Randomized rumor spreading'', {\em Proceedings of the IEEE Conference on Foundations of Computer Science (FOCS'00),} 2000.

\bibitem{Kempe}
D. Kempe, J. M. Kleinberg, and A. J. Demers, 
``Spatial gossip and resource location protocols,'' 
{\em in Proc. ACM Symp. Theory of Computing,} Heraklion, Crete, Greece, Jul. 2001.

\bibitem{Trickle}
P.~Levis, N.~Patel, D.~Culler, and S.~Shenker, ``Trickle: A self-regulating algorithm for code propagation and maintenance in wireless sensor networks'', {\em Proceedings of the USENIX NSDI Conference,} San Francisco, March 2004.

\bibitem{Mesbahi}
M.~Mesbahi and M.~Egerstedt, {\em Graph Theoretic Methods in Multiagent Networks,} Princeton University Press, 2010.

\bibitem{Mosk-Aoyama} 
D. Mosk-Aoyama and D. Shah, 
``Fast distributed algorithms for computing separable functions'', 
{\em IEEE Trans. Inform. Theory,} vol.~54, no.~7, July~2008, pp.~2997--3007.
 
\bibitem{Reingold}
O.~Reingold, S.~Vadhan, and A.~Widgerson, ``Entropy waves, the zig-zag graph product, and new constant-degree expanders'', {\em Annals of Mathematics,} vol.~155, no.~2, 2002, pp.~157--187.


\end{thebibliography}
\end{document}